\documentclass[12pt]{article}

\usepackage[cp1251]{inputenc}

\usepackage{amsfonts}
 \usepackage{amsthm}
 \usepackage{amsmath}
 \usepackage{amscd}
 \usepackage{amssymb}

\usepackage[english]{babel}
\usepackage{tikz-cd}

 \usepackage{graphics}
\usepackage{graphicx}

 \numberwithin{equation}{section}
 \newtheorem{theorem}{Theorem}
 \newtheorem{lm}{Lemma}
 
 \theoremstyle{definition}

\title{Semi-global symplectic invariant of the champagne bottle}
\author{Ognyan Christov\\
Faculty of Mathematics and Informatics, \\
Sofia University "St. Kliment Ohridski", \\
5 J. Bouchier blvd., 1164 Sofia, Bulgaria}
\date{}

\begin{document}

 \maketitle

\begin{abstract}

\noindent We study  a two degrees of freedom Hamiltonian system
describing the motion of a particle in a potential field of the
form of $S^1$ symmetric double well, namely $V = - (x_1^2 + x_2^2)
+ (x_1^2 + x_2^2)^2$, known also as a champagne bottle potential.
 This system is completely integrable. The
champagne bottle is the simplest member of a class of integrable
systems that have no global action variables due to a non-trivial
monodromy, Bates (1991). Beyond that, the geometric and dynamical
properties of the system  near the equilibrium are of primary
interest. We calculate the Birkhoff normal form and the nontrivial
action near the focus-focus singularity and obtain the semi-global
symplectic invariant near focus-focus point, which is introduced
by V\~{u} Ng\d{o}c (2003). Examples of such calculations are
still few. We compare our result with the semi-global
symplectic invariant of the spherical pendulum, calculated by Dullin (2013).
\end{abstract}


\section{Introduction}

In this paper we deal with the Hamiltonian
\begin{equation}
\label{1.1}
H = \frac{1}{2} (y_1 ^2 + y_2^2) - (x_1^2 + x_2^2) + (x_1^2 + x_2^2)^2,
\end{equation}
defined on the phase space $M^4 = \{ (x_1, x_2, y_1,y_2) \in
\mathbb{R}^4 \}$  endowed with the standard symplectic form
$\Omega = d x_1 \wedge d y_1 + d x_2 \wedge d y_2$ and  the
standard Poisson structure $\{x_i, y_j \} = \delta_{i,j}$. The
Hamiltonian system associated with $H$
\begin{eqnarray}
\label{1.2}
& \dot{x}_1 = y_1, \quad & \dot{y}_1 = 2 x_1 [1 - 2 (x_1 ^2 + x_2^2)],  \\
& \dot{x}_2 = y_2, \quad & \dot{y}_2 = 2 x_2 [1 - 2 (x_1 ^2 +
x_2^2)] \nonumber
\end{eqnarray}
is completely integrable due to the $S^1$ symmetry. The second
integral (the angular momentum) reads
\begin{equation}
\label{1.3}
J_2 = x_1 y_2 - x_2 y_1 .
\end{equation}
This system describes the motion of a particle in a champagne
bottle potential, also referred as  a particle in a Mexican hat
potential. Such potentials are of interest to field theorists
studying the Higgs field. As can be seen, this system does not
depend on any parameters. This makes it the simplest member of a
class of systems that exhibit Hamiltonian monodromy, which is an
obstruction for the existence of global action variables Bates
\cite{B}. It turns out that this topological feature causes a
sharp change in the disposition of the quantum-mechanical spectrum
near the zero energy, see  Child \cite{Child}.

We are interested in the geometry and the dynamical properties of
the system (\ref{1.2}) near the equilibrium $(0, 0)$, and mainly
in the computation of the semi-global symplectic invariant near
this point, which turns out to be a focus-focus point.

In the pursuit of solving the classification problem of integrable
systems with singularities V\~{u} Ng\d{o}c \cite{Ngoc} introduced
semi-global symplectic invariants for the Liouville foliation of
integrable Hamiltonian systems with two degrees of freedom near
the separatrix of a simple focus-focus point. This result is an
extension of the invariants introduced by Dufour et al \cite{DMT}
for one degree of freedom systems with simple hyperbolic
singularity.

Prior to that topological and smooth invariants are introduced,
see \cite{BoFo,Zung0} and also the survey article \cite{BoOsh}.

Later, Dullin and V\~{u} Ng\d{o}c \cite{DSVN} introduced the
semi-global symplectic invariants for hyperbolic-hyperbolic
equilibrium for two degrees of freedom systems and computed this
invariant for the Neumann system.

In our treatment of the considered system we follow closely the
approach taken by Dullin \cite{Dullin} in his study of semi-global
symplectic invariants for the spherical pendulum. It is also worth
mentioning the work of Alonso et all \cite{ADH} in which some
invariants for the spin-oscillator Hamiltonian integrable system
(important example in the theory of semi-toric systems) are
calculated, see also \cite{A,BD}.

The investigated here Hamiltonian system has a simple focus-focus
point and it is genera\-li\-zed semi-toric system, since the
angular momentum integral fails to be proper (the proof of that
fact is exactly the same as in Pelayo et all \cite{PRNg}).

The paper is organized as follows. All necessary concepts and
results are recalled in Section 2 In Section 3 the Birkhoff normal
form near the focus-focus point is calculated. In Section 4 we
compute the expansion of the non-trivial action integral near
focus-focus singularity. Then we find the semi-global symplectic
invariant as well as the other characteristics as the period, the rotation number and the twist
in Section 5. We recover  some results about the KAM conditions for this system.
At the end, we compare the semi-global symplectic invariant of the champagne bottle
with that of the spherical pendulum and conclude that
the spherical pendulum and the champagne bottle are not symplectically equivalent.

\section{Theoretical background}

In this section we recall well known results about integrable
Hamiltonian systems. Some of them are valid for systems with
arbitrary number of degrees of freedom, some of them are valid
only for the systems with two degrees of freedom and are in a
process of generalization (see, for instance,
\cite{AbMa,Arnold,BoFo}).

A Hamiltonian system related to a Hamiltonian $H$, defined on a
symplectic manifold $(M^{2n}, \Omega)$ is called completely
integrable if there exist $n$ first integrals $f_1 = H, f_2,
\ldots, f_n$, functionally independent almost everywhere and
Poisson commuting, that is, $\{f_i, f_j\} = 0$ for all $i, j$.
These first integrals define the momentum map $F:= (f_1, f_2,
\ldots, f_n) : M^{2n} \to \mathbb{R}^n$.

A point $m \in M^{2n}$ is called a regular point if $rank DF = n$
and a singular point if $rank DF < n$. A value $a \in
\mathbb{R}^n$ is called regular if $F^{-1} (a)$ consists only of
regular points and critical otherwise. Equivalently independency
of the integrals almost everywhere means that $rank DF = n$ almost
everywhere.

Let $m \in M^{2n}$ be an equilibrium point such that $d f_1 (m) =
d f_2 (m) = \ldots = d f_n (m)$ ($rank DF = 0$). Consider $A_j =
(\Omega)^{-1} d^2 f_j (m)$ as a linear operator: $A_j \in sp
(T_m M, \Omega)$. Since $f_1, f_2, \ldots, f_n$ commute, $A_1,
A_2, \ldots, A_n$ generate a commutative subalgebra $K$ in $sp
(T_m M, \Omega)$. The point $m \in M^{2n}$ is {\it
non-degenerate}, if $K$ is a Cartan subalgebra.

The map $F$ defines a singular Liouville fibration on $M^{2n}$,
whose fibers are connected components of $F^{-1} (c), c \in
\mathbb{R}^n$ which we assume to be compact. A fibre $F^{-1} (c)$
is called regular if $c$ is a regular value and singular if it
contains at least one singular point.

{\bf Remark 1.} Clearly, for a given integrable Hamiltonian $H$,
the moment map is not unique. However, under certain non-resonance
condition (see the end of this section), the regular level sets of
the moment map are uniquely determined by the system.

If $c$ is a regular value of $F$, we have near any compact
connected component $\Lambda_c$ of $F^{-1} (c)$ the
Liouville-Arnold theorem \cite{AbMa,Arnold} which says that
$\Lambda_c$ is diffeomorphic to an $n$-dimensional torus
$\mathbb{T}^n$. The neighborhood of this torus is symplectomorphic
to the standard model $ \mathbb{T}^n \times D^n$ with $\Omega =
\sum d \varphi_j \wedge d I_j$, where ${\boldsymbol{\varphi}} =
(\varphi_1, \varphi_2, \ldots, \varphi_n)$ are coordinates on
$\mathbb{T}^n$ called {\it angles} and $\mathbf{I} = (I_1, I_2,
\ldots, I_n)$ are coordinates on $D^n$ called {\it actions}. The
dynamics is governed by extremely simple equations
\begin{equation}
\label{he}
\dot{\mathbf{I}} = 0, \quad \dot{\boldsymbol{\varphi}}
= \frac{\partial H}{\partial \mathbf{I}}, \quad H = H
(\mathbf{I}).
\end{equation}
The actions can be computed as follows
\begin{equation}
\label{actionvar}
I_j = \frac{1}{2 \pi} \int_{\gamma_j} \delta, \qquad j = 1, \ldots, n,
\end{equation}
where $\delta$ is the Liouville 1-form $d \delta = \Omega$ and
$\gamma_j$ are cycles on $\mathbb{T}^n$.

In general, the actions are not globally defined due to some
geometric obstructions as the Hamiltonian monodromy, see
Duistermaat \cite{Dui}. The considered in this paper system is an
example of an integrable system without global actions.

The action variables are the most natural symplectic invariants of
integrable systems. In other words, the Liouville-Arnold theorem
gives a classification of integrable systems (up to the dimension)
locally near regular fibres and all such systems look the same.

However, the studied Hamiltonian systems do have singularities.
Therefore, the dynamics near singular fibers have to be studied to
find distinctions between integrable systems. In this way a
general problem arises: to describe invariants and to classify
singular Liouville fibrations.

We mentioned some works about the different type of
classifications of non-degenerate singularities in the
Introduction. Now, we give a more detailed description of the
symplectic invariant for focus-focus singularities. To do that, we
will consider integrable Hamiltonian systems with two degrees of
freedom, namely there are two independent, commuting first
integrals $f_1$ and $f_2$.

The focus-focus singularities are one of the four types of
singularities of Morse-Bott type in dimension 4 \cite{Eliasson,Will}.
A fixed point $m \in M^4$ ($d f_1 (m) = d f_2 (m) = 0$) is called
focus-focus point if there are symplectic coordinates $(q_i, p_i)$
near $m$ (that is, $\Omega = d q_1 \wedge d p_1 + d q_2 \wedge d
p_2$) such that
$$
f_1 = a (q_1 p_1 + q_2 p_2) + b (q_1 p_2 - q_2 p_1) + h.o.t. ,
$$
$$
f_2 = c (q_1 p_1 + q_2 p_2) + d (q_1 p_2 - q_2 p_1) + h.o.t. ,
$$
where $a, b, c, d$ are constants satisfying $a d - b c \neq 0$.
Equivalently, the focus-focus equilibrium has eigenvalues $\pm \nu
\pm i \omega$, where $\nu, \omega$ are non-zero real numbers.
Such eigenvalues are also called loxodromic.

Observe that, at least for the systems with two degrees of
freedom, the focus-focus singularities and non-trivial monodromy
are connected \cite{B,Dullin,Dui,Zung2}

Further, it is also assumed that the singular Liouville fibration
has a unique singular point $m$, which is of focus-focus type and
the fiber containing it is compact -- simple focus-focus
singularity.

The fibration is defined by the two integrals $f_1$ and $f_2$.
They can be chosen in such a way that the actions are
\begin{align}
I_1 &= \mathrm{Re} ((f_1 + i f_2) \ln (f_1 + i f_2) ) + c (f_1, f_2), \nonumber \\
I_2 &= f_2,
\end{align}
where $c (f_1, f_2)$ is smooth.

\begin{theorem}
\label{Th1}
 (San V\~{u} Ng\d{o}c  \cite{Ngoc}, 2003) The Taylor expansion of $c (f_1, f_2)$ is
a symplectic invariant of the simple focus-focus singularity.
\end{theorem}

We will make the expression of $c (f_1, f_2)$ more precise in
Section 5. It is seen that the  symplectic invariant is closely
related to the actions as in the regular case. Moreover, one can
consider $c (f_1, f_2)$ as the regularized part of the action.

As an important result, the introduced invariant classifies the
singular Liouville fibration in an open vicinity of the singular
fibre, up to a symplectomorphism.

We finish this section with some remarks about the KAM theory
conditions for the champagne bottle. In 1995 Georgiev \cite{GG}
showed that the Kolmogorov condition \cite{Arnold}
$$
\det (\partial^2 H/\partial I_i \partial I_j) \neq 0,
$$
where $H$ is an integrable Hamiltonian and $I_j$ are the action
variables, is fulfilled everywhere out of the bifurcation diagram
of the momentum map using the complex-analytic approach developed
by Horozov \cite{EH1,EH2}. In general, if the Hamiltonian of a given
integrable system $H$ satisfies the Kolmogorov condition (almost
everywhere), then $H$ is non-resonant.

Georgiev also showed that the
iso-energetic non-degeneracy condition (twist condition)
$$
\det
\begin{pmatrix}
\frac{\partial^2 H}{\partial I^2} & \frac{\partial H}{\partial I} \\[1ex]
\frac{\partial H}{\partial I} & 0
 \end{pmatrix} \neq 0
$$
is violated along the curve through the focus-focus point (see the dashed curve in Fig. 1).

A year later, Zung \cite{Zung1} proved that for an integrable two
degrees of freedom Hamiltonian system with a simple focus-focus
value, the Kolmogorov condition is always satisfied in a
neighborhood of the focus-focus singularity.

Several years later, Dullin and  V\~{u} Ng\d{o}c \cite{DSVN} have
improved the Zung's result using the theory developed in
\cite{Ngoc} and described above. Their results show that in a neighborhood of  a simple focus-focus point,
the Arnold-Moser condition is violated while the Kolmogorov condition is satisfied.

At almost same time, Rink \cite{Rink} presented a new proof of the Zung's result with a
quite similar to \cite{DSVN} technique.

\section{Normal Form}

In this section we will compute the Birkhoff normal form for the
Hamiltonian (\ref{1.1}) near the equilibrium. The  Birkhoff normal
form is a simplified Hamiltonian and that simplification is
achieved by series of near-identity canonical transformations
obtained by a generating function. For this purpose, the Lie
series approach developed by Deprit \cite{Deprit} will be used,
which we will recall at the beginning briefly (see for full
details also \cite{MeyerHall}).

Suppose without loss of generality that the equilibrium is $(0,
0)$. The scaling $(x, H) \to (\varepsilon x, \varepsilon^{-2} H)$
"zooms in" to the equilibrium. Here $\varepsilon$ is a formal
parameter, which we can set $\varepsilon=1$ eventually.

We are given a Hamiltonian in the form
\begin{equation}
\label{2.1}
H (x, \varepsilon) = \sum_{i=0} ^{\infty} \frac{\varepsilon^i}{i!} H_i ^0 (x)
\end{equation}
Via symplectic, near-identity change of variables $x \to y$, which
does not transforms $\varepsilon$, the above Hamiltonian becomes
\begin{equation}
\label{2.2}
\overline{H} (y, \varepsilon) = \sum_{i=0} ^{\infty} \frac{\varepsilon^i}{i!} H_0 ^i (y)
\end{equation}
Usually $\overline{H}$ is called the Lie transform of $H$,
generated by $W$
\begin{equation}
\label{2.3}
W (x, \varepsilon) = \sum_{i=0} ^{\infty} \frac{\varepsilon^i}{i!} W_{i+1} (x)
\end{equation}
The new Hamiltonian has to be simpler than the original, so it is
said to be in normal form. Notice that this normal form need not
to be convergent.

To facilitate the calculations the intermediate Hamiltonians
$\{H_j ^i \}$, $i = 1, 2, \ldots, j = 0, 1, \ldots$ are introduced
by the recursive identities
\begin{equation}
\label{2.4}
H_j ^i = H_{j+1} ^{i-1} + \sum_{k=0} ^j  \binom{j}{k} \{H_{j-k} ^{i-1}, W_{k+1} \}
\end{equation}
The algorithm starts with a given Hamiltonian, that is, all $H_i
^0$ are known (we follow closely \cite{MeyerHall}). Assume all the
rows of the Lie-Deprit triangle

\[
  \begin{tikzcd}
H_0 ^0  \arrow{d}    \\
H_1 ^0  \arrow{d} \arrow{r} & H_0 ^1 \arrow{d} \\
H_2 ^0  \arrow{d} \arrow{r} & H_1 ^1 \arrow{d} \arrow{r} & H_0 ^2 \arrow{d} \\
H_3 ^0  \arrow{d} \arrow{r} & H_2 ^1 \arrow{d} \arrow{r} & H_1 ^2 \arrow{d} \arrow{r} & H_0 ^3 \arrow{d}\\
   {}  &  {}  & {} & {}
  \end{tikzcd}
\]

\vspace{2ex}

\noindent have been computed down to the $(N-1)$st row, thus $W_1,
\ldots, W_{N-1}$ have been determined and the desired members of
the normal form $H_0 ^1, \ldots, H_0^{N-1}$ are obtained. To
compute the $N$th row we perform the following steps:

\vspace{1ex}

Step 1. Assume $W_N = 0$ and compute all the terms of $N$th row
using (\ref{2.4}) and denote them by $L_j ^i, i+j=N$.

\vspace{1ex}

Step 2. Solve the equation $H_0 ^N = L_0 ^N + \{H_0 ^0, W_N\}$ for
$W_N$ and $H_0 ^N$, so that $H_0 ^N$ is in normal form.

\vspace{1ex}

Step 3. Calculate $H_j ^i = L_j ^i + \{H_0 ^0, W_N\}$ for all $i +
j = N$.

\vspace{1ex}

Step 4. Repeat for the next row.

\vspace{1ex}

The definition of normal form depends on the equation $H_0 ^N =
L_0 ^N + \{H_0 ^0, W_N\}$, which in turn depends on $H_0 ^0$. This
equation is called the homological or the Lie equation.

This algorithm can be used to establish the following

\begin{theorem} (Theorem 10.3.1, \cite{MeyerHall})
\label{Th2} Let $\{\mathcal{P}_i \}_{i=0} ^{\infty}$,
$\{\mathcal{Q}_i \}_{i=1} ^{\infty}$ and $\{\mathcal{R}_i \}_{i=1}
^{\infty}$ be sequences of linear spaces of smooth functions
defined on a common domain $O$ in $\mathbb{R}^n$ with the
following properties:

1. $\mathcal{Q}_i \subset \mathcal{P}_i$, \quad $i = 1, 2, \ldots$

2. $H_i ^0 \in \mathcal{P}_i$, \quad $i = 0, 1, 2, \ldots$

3. $\{\mathcal{P}_i, \mathcal{R}_j \} \subset \mathcal{P}_{i+j}$,
\quad $i+j = 1, 2, \ldots$

4. for any $D \in \mathcal{P}_i, i = 1, 2, \ldots$, there exists
$B \in \mathcal{Q}_i$ and $C \in \mathcal{R}_i$, such that
\begin{equation}
\label{2.5} B = D + \{H_0 ^0, C \}.
\end{equation}
Then, there exists a function $W$ of the form (\ref{2.3}) with
$W_i \in \mathcal{R}_i, i = 1, 2, \ldots$, which generates a
near-identity symplectic change of variables $x \to y$ such that
the Hamiltonian in the new variables has a series expansions given
by (\ref{2.2}) with $H_0 ^i \in \mathcal{Q}_i, i = 1, 2, \ldots$.
\end{theorem}

\noindent For the uniqueness of the normal form one needs two more
conditions (Theorem 10.3.3 \cite{MeyerHall}):
\begin{itemize}

\item the linear operator $\mathcal{H}_i = \{H_0 ^0, .\}: \mathcal{P}_i \to \mathcal{P}_i$ to be simple, that is,
$$
\mathcal{P}_i = {\rm kernel} (\mathcal{H}_i) \oplus {\rm range} (\mathcal{H}_i)
$$
and
\item
$ \{\mathcal{Q}_i, \mathcal{Q}_j \} = 0, \quad i, j = 1, 2, \ldots . $
\end{itemize}

Returning to our system, we define the momentum mapping $F = (H, J_2)$ on some open subset $V$ of $M^4$ with values in $\mathbb{R}^2$
\begin{equation}
\label{2.55}
F : V \to (h, j_2) \in \mathbb{R}^2 .
\end{equation}

Next, we find that the equilibrium points of (\ref{1.2}) are $(0,
0)$ and $y_1 = y_2 = 0, x_1 ^2 + x_2 ^2 = \frac{1}{2}$. We are
interested mainly in the dynamics around  $(0, 0)$. The
linearization around $(0, 0)$ gives eigenvalues $\pm \sqrt{2}$
with multiplicity two. Therefore, this equilibrium is a degenerate
saddle-saddle point. On the other hand, from the perspective of
the foliation of the integrable system, it is a focus-focus point,
see Zung \cite{Zung2}. Clearly, any function of the first
integrals $H$ and $J_2$ is a first integral, and hence, has the
same foliation. Then, the linear combination $H + \omega J_2$ for
arbitrary $\omega \neq 0$ has an equilibrium of focus-focus type with
eigenvalues $\pm \sqrt{2} \pm i \omega$.

Further, we want to transform the quadratic part of the
Hamiltonian (\ref{1.1}) into Williamson normal form \cite{Will}.
The symplectic rescaling (preserves $\Omega$)
\begin{equation}
\label{2.6}
x_{1,2} \to \frac{1}{\sqrt[4]{2}} x_{1,2}, \qquad y_{1,2} \to \sqrt[4]{2} y_{1,2}
\end{equation}
brings (\ref{1.1}) to the form
\begin{equation}
\label{2.7}
H = \frac{\sqrt{2}}{2}\left( y_1^2 - x_1 ^2 \right) +
\frac{\sqrt{2}}{2}\left( y_2^2 - x_2 ^2 \right) + \frac{1}{2}
(x_1^2 + x_2^2)^2.
\end{equation}
Further, we make another symplectic change
\begin{equation}
\label{2.8}
x_j = \frac{1}{\sqrt{2}} (q_j - p_j) , \qquad y_j = \frac{1}{\sqrt{2}} (q_j + p_j), \qquad j = 1, 2.
\end{equation}
The symplectic form becomes
\begin{equation}
\label{2.9}
\Omega = d q_1 \wedge d p_1 + d q_2 \wedge d p_2 .
\end{equation}
Denote also
\begin{equation}
\label{2.10}
J_1 = p_1 q_1 + p_2 q_2.
\end{equation}
In these coordinates the second integral  reads $J_2 = q_1 p_2 -
q_2 p_1$ and the Hamiltonian (\ref{2.2}) is
\begin{align}
\label{2.11}
H  = &  H_2 + H_4,  \\
& H_2 = \sqrt{2} J_1, \quad H_4 = \frac{1}{8} \left( q_1^2 + q_2^2
+ p_1^2 + p_2^2 - 2 J_1 \right)^2. \nonumber
\end{align}

Scaling $(q, p, H) \to (\varepsilon q, \varepsilon p,
\frac{H}{\varepsilon^2})$, where $\varepsilon$ is a formal small
parameter which we can set $\varepsilon = 1$ at the end. Then the
Hamiltonian reads
\begin{align}
\label{2.12}
H = & \, H_0 ^0 + \frac{\varepsilon^2}{2} H_2 ^0, \qquad H_1 ^0 \equiv H_j ^0 \equiv 0, \quad j > 2, \\
    & H_0 ^0 = \sqrt{2} J_1, \qquad H_2 ^0 = \frac{1}{4} (q_1^2 + q_2^2 + p_1^2 + p_2^2 - 2 J_1)^2 \nonumber
\end{align}
Following Dullin \cite{Dullin} we introduce "almost action-angle"
coordinates by
\begin{align}
\label{2.13}
q_1 &= e^{\theta_1} \cos \theta_2, &p_1 = (J_1 \cos \theta_2 - J_2 \sin \theta_2)e^{-\theta_1} , \\
q_2 &= e^{\theta_1} \sin \theta_2, &p_2 = (J_1 \sin \theta_2 + J_2
\cos \theta_2)e^{-\theta_1}. \nonumber
\end{align}
Since
\begin{equation}
\label{2.14}
\theta_1 = \ln \sqrt{q_1^2 + q_2^2}, \qquad \theta_2 = \arctan \frac{q_2}{q_1}
\end{equation}
these coordinates are symplectic when $(q_1, q_2) \neq (0, 0)$,
$$
\Omega = d \theta_1 \wedge d J_1 + d \theta_2 \wedge d J_2, \qquad \{\theta_i, J_j \} = \delta_{ij}
$$
and only $\theta_2$ is an angle.

With these variables $H_2 ^0$ becomes
\begin{equation}
\label{2.15}
H_2 ^0 = \frac{1}{2} (3 J_1^2 + J_2^2) +
\frac{1}{4}\left(e^{4\theta_1} + e^{-4\theta_1}(J_1^2 + J_2^2)^2
-4 J_1 e^{2 \theta_1} - 4 J_1 (J_1^2 + J_2^2) e^{-2 \theta_1}
\right).
\end{equation}
Notice that by construction $H_0 ^i$ are functions of $(J_1,
J_2)$, since $J_1, J_2$ generate the algebra of functions that
commute with $H_0 ^0$. Furthermore, the original Hamiltonian is
invariant under the discrete symmetry $J_2 \to - J_2$. It can be
easily turned into a discrete symplectic symmetry $S: (J_2,
\theta_2) \to - (J_2, \theta_2)$. Due to a Theorem from Gaeta
\cite{Gaeta} the normal form also enjoys this symmetry, that is,
it depends on $J_2 ^2$.

Define the class of functions $\mathcal{P}_i, i > 0$
$$
\sum_{k=0} ^{2i} P_k (J_1, J_2) e^{2(i-k)\theta_1},
$$
where $P_k (J_1, J_2)$ is a homogeneous polynomial of degree $k$.
The operator $\mathcal{H}_i  = \{H_0 ^0, .\}: \mathcal{P}_i \to
\mathcal{P}_i$ just gives $W_i \to -\sqrt{2} \partial W_i /
\partial \theta_1$. Moreover, this operator is simple:
$\mathcal{P}_i = {\rm kernel} (\mathcal{H}_i) \oplus {\rm range}
(\mathcal{H}_i)$. Obviously, the kernel consists of functions
$\mathcal{Q}_i$, which are independent of $\theta_1$. The range
consists of functions $\mathcal{R}_i = \mathcal{P}_i \setminus
\mathcal{Q}_i$.

For the uniqueness we need to verify two more conditions:
$\{\mathcal{P}_i, \mathcal{R}_j \} \subset \mathcal{P}_{i+j}$ and
$ \{\mathcal{Q}_i, \mathcal{Q}_j \} = 0$. The second condition is
fulfilled, because $\{J_1, J_2\} = 0$. For the first condition we
have
$$
\{\mathcal{P}_i, \mathcal{R}_j \} \subset \{\mathcal{P}_i,
\mathcal{P}_j \} \subset \mathcal{P}_{i+j - 1} \subset
\mathcal{P}_{i+j}.
$$
Now, we apply Theorem \ref{Th1} and the iterative procedure given
above. We give the first steps of the calculation of the normal
form, which are easy.

From (\ref{2.4}) it follows
\begin{equation}
\label{2.16}
H_0 ^1 = H_1 ^0 + \{H_0 ^0, W_1\}.
\end{equation}
Since $H_1 ^0 \equiv 0$, we choose $H_0 ^1 = W_1 \equiv 0$. Next,
\begin{equation}
\label{2.18}
H_0 ^2 = H_1 ^1 + \{H_0 ^0, W_1\} = H_1 ^1.
\end{equation}
Again from (\ref{2.4}) for $H_1 ^1$ we have
\begin{equation}
\label{2.19}
H_1 ^1 = H_2 ^0 + \{H_0 ^0, W_2\} = H_2 ^0 + \sqrt{2}
\{J_1, W_2 \} = H_2 ^0 - \sqrt{2} \frac{\partial W_2}{\partial
\theta_1}.
\end{equation}
Combining the last two identities yields
\begin{align}
\label{2.20}
H_0 ^2 = H_1 ^1 = & \frac{1}{2} (3 J_1^2 + J_2^2) \\
                + & \frac{1}{4}\left(e^{4\theta_1} + e^{-4\theta_1}(J_1^2 + J_2^2)^2
-4 J_1 e^{2 \theta_1} - 4 J_1 (J_1^2 + J_2^2) e^{-2 \theta_1}
\right) - \sqrt{2}\frac{\partial W_2}{\partial \theta_1}.
\nonumber
\end{align}
We choose
\begin{equation}
\label{2.21}
H_0 ^2 = \frac{1}{2} (3 J_1^2 + J_2^2)
\end{equation}
and then we solve the homological equation
\begin{equation}
\label{2.22}
\sqrt{2}\frac{\partial W_2}{\partial \theta_1} =
\frac{1}{4}\left(e^{4\theta_1} + e^{-4\theta_1}(J_1^2 + J_2^2)^2
-4 J_1 e^{2 \theta_1} - 4 J_1 (J_1^2 + J_2^2) e^{-2 \theta_1}
\right)
\end{equation}
for $W_2$ after trivial integration to get
\begin{equation}
\label{2.23}
W_2 = \frac{1}{4 \sqrt{2}} \left( \frac{1}{4}
e^{4\theta_1} - 2 J_1 e^{2 \theta_1} + 2 J_1 (J_1^2 + J_2^2) e^{-2
\theta_1} - \frac{1}{4} e^{-4\theta_1}(J_1^2 + J_2^2)^2 \right).
\end{equation}
To perform next steps the computer algebra is used to calculate
numerous brackets. The convergency is not an issue, since it is
well-known that for Liouville integrable systems this normal form
is convergent \cite{Ito}. Skipping the details, we obtain
\begin{theorem}
\label{ThNF}
The normal form at the focus-focus point $(0, 0)$ of the champagne bottle is
\begin{align}
\label{2.24}
\overline{H} = &\sqrt{2} J_1 + \frac{1}{4} (3J_1^2 +
J_2^2) - \frac{\sqrt{2}}{32} J_1 (17 J_1^2 + 9 J_2^2)
 + \frac{1}{256} (375 J_1^4 + 258 J_1^2 J_2^2 + 11 J_2^4)  \\
& - \frac{\sqrt{2}}{4096} J_1 (10689 J_1^4 + 8910 J_1^2 J_2^2 +
909 J_2^4) + O (J^6). \nonumber
\end{align}
\end{theorem}

For certain reasons further on we will need the inverse of
$\overline{H} (J_1, J_2) = h$ with respect to $J_1$  $(J_2 =
j_2)$, so that $\overline{H} (J_1 (h, j_2) , j_2) = h$ is an
identity. Straightforward computations yield that
\begin{equation}
\label{2.25}
J_1 (h, j_2) = \frac{1}{\sqrt{2}} \Big[ h -
\frac{1}{8}(3 h^2 + 2 j_2 ^2) + \frac{h}{64} (35 h^2 + 30 j_2^2) -
\frac{5}{1024}(231 h^4 + 252h^2 j_2^2 + 28 j_2^4)\Big] + \ldots
\end{equation}

\section{Action Integrals}

In this section, the calculation of the action integrals for the
champagne bottle is carried out. First, we have to find where they
exist. Since the already used coordinates in the previous section
do not fit well for that purpose, we introduce the polar
coordinates
\begin{align}
\label{3.1}
x_1 &= r \cos \varphi, &y_1 = p_r \cos \varphi - \frac{p_r}{r} \sin \varphi, \\
x_2 &= r \sin \varphi, &y_2 = p_r \sin \varphi + \frac{p_r}{r}
\cos \varphi. \nonumber
\end{align}
Then the Hamiltonian (\ref{1.1}) takes the form
\begin{equation}
\label{3.2}
H = \frac{1}{2} \left(p_r ^2 + \frac{p_{\varphi} ^2}{r^2} \right) + r^4 - r^2
\end{equation}
and the other integral becomes $J_2 = p_{\varphi} = j_2$.

It is shown in Bates \cite{B} that the critical values of the
momentum mapping $F$ (\ref{2.55}) of the champagne bottle are $(0, 0)$ and the curve parameterized by
$$
(h, j_2) := (3 r^4 - 2 r^2, \pm \sqrt{4 r^6 - 2 r^4}), \qquad r \geq \frac{1}{\sqrt{2}}.
$$
Denote the set of regular values of $F$ by $U_r$, see Fig. 1. For
the points $(h, j_2) \in U_r$ the connected components of $ F^{-1}
(h, j_2)$ (that is, the connected components of the level surfaces
determined by the equations $H = h, J_2 = j_2$) are two - tori
$T_{h, j_2}$.

\begin{figure}[ht]
\centering
\begin{tikzpicture}
\draw [help lines, ->] (-3.5,0)--(3.5,0); \draw [help lines, ->]
(0,-3.1)--(0,3.0);

\node [below right] at (3.5,0) {$j_2$}; \node [left] at (0,3.0)
{$h$};

\draw (-3,2.5) .. controls (0,-4)  .. (3,2.5);

\node [below left] at (2,2) {$U_r$}; \node [below left] at
(0,-2.5) {$-1/4$};

\draw[dashed](-2.5,1.5) .. controls (0,-0.5)  .. (2.5,1.5);

\draw[red,fill=red] (0,0) circle (.5ex);

\end{tikzpicture}
\caption{The set $U_r$ of regular values of the momentum map.}
\end{figure}
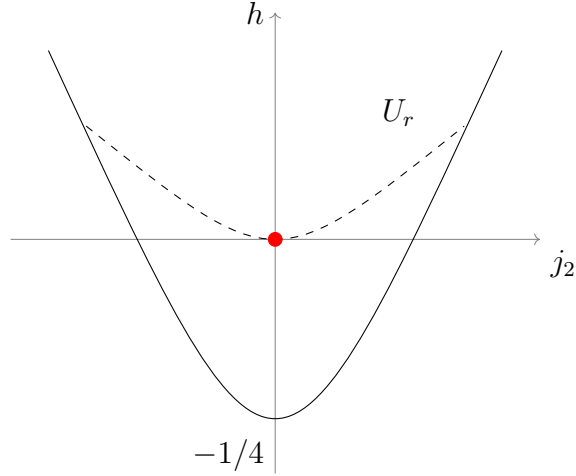

Let us choose a basis of the homology group $\rm{H}_1 (T_{h, j_2},
\mathbb{Z})$ with the following representatives: for $\gamma_1$ we
fix $p_{\varphi}$ and $\varphi$ and let $r, p_r$ make one circle
on the curve
$$
\frac{1}{2} \left(p_r ^2 + \frac{p_{\varphi} ^2}{r^2} \right) +
r^4 - r^2 = h;
$$
for $\gamma_2$ we fix $r, p_r$ and make $\varphi$ to cover $ [0, 2
\pi]$. Then the action integrals are
\begin{equation}
\label{3.3}
I_1 = \frac{1}{2\pi} \oint_{\gamma_1} p_r d r =
\frac{1}{2\pi} \oint_{\gamma_1} \frac{\sqrt{2(hr^2 + r^4 - r^6) -
j_2^2}}{r} d r.
\end{equation}
\begin{equation}
\label{3.4}
I_2 = J_2 = \frac{1}{2\pi} \oint_{\gamma_2} p_{\varphi} d \varphi = j_2
\end{equation}
It is convenient to scale the action integrals as
\begin{equation}
\label{scal}
\tilde{I}_1 = \sqrt{2} I_1, \quad \tilde{J}_1 = \sqrt{2} J_1, \quad \tilde{J}_2 = \sqrt{2} J_2, \quad \tilde{j}_2 = \sqrt{2} j_2
\end{equation}
and drop the tildes thereafter. We can always go back from
(\ref{scal}) when necessary. Observe that, in view of the above
scaling the normal form (\ref{2.24}) and its inverse (\ref{2.25}) become
\begin{align}
\label{nf}
\overline{H} = &J_1 + \frac{1}{8} (3J_1^2 + J_2^2) - \frac{1}{64} J_1 (17 J_1^2 + 9 J_2^2) + \frac{1}{1024} (375 J_1^4 + 258 J_1^2 J_2^2 + 11 J_2^4)  \\
& - \frac{1}{16 384} J_1 (10689 J_1^4 + 8910 J_1^2 J_2^2 + 909J_2^4) + \ldots = h, \nonumber
\end{align}
\begin{equation}
\label{invnf}
J_1 (h, j_2) =  h - \frac{1}{8}(3 h^2 +  j_2 ^2) + \frac{h}{64} (35 h^2 + 15 j_2^2) -
\frac{1}{1024}(1155 h^4 + 630h^2 j_2^2 + 35 j_2^4) + \ldots .
\end{equation}

Further, we transform (\ref{3.3}) by putting
\begin{equation}
\label{3.5}
z = r^2, \quad d z = 2 r d r, \quad w^2 = P (z):= - 2 z^3 + 2 z^2 + 2 h z - \frac{j_2 ^2}{2}.
\end{equation}
Denote the real oval of the elliptic curve
\begin{equation}
\label{3.6}
\Gamma = \left\{ (z, w): \quad w^2 = P (z) = - 2 z^3 + 2 z^2 + 2 h z - \frac{j_2 ^2}{2} \right\},
\end{equation}
which exists for all $(h, j_2) \in U_r$ by $\beta$. Then (\ref{3.3}) becomes
\begin{equation}
\label{3.7}
I_1 (h, j_2) = \frac{\sqrt{2}}{4 \pi} \oint_{\beta} \frac{w}{z} d z.
\end{equation}
Let the roots of $P (z)$ be $z_j, \, j = 1, 2, 3$. They are real
and located as follows: $z_1 \leq 0 \leq z_2 < z_3$. The
$\beta$-cycle encloses the interval $[z_2, z_3]$ along which $w^2
\geq 0$, see Fig. 2. On the other hand, the $\alpha$ - cycle
encloses the interval  $[z_1, z_2]$ along which $w^2 \leq 0$. That
is why it is  called an {\it imaginary} cycle. For a nice
introduction of the {\it imaginary} cycles of the elliptic curves,
see for example \cite{A}. Evidently, when $h, j_2 \to 0$,  then
$z_1, z_2 \to 0$, so this cycle becomes arbitrary small, therefore
it is also called {\it vanishing} cycle.

\vspace{2ex}

\begin{figure}[ht]
\centering
\begin{tikzpicture}
\draw [help lines, ->] (-3.5,0)--(8.5,0);

\draw [blue, line width=0.8] (4.15,0) ellipse (2.15 and 0.8);
\node [above right] at (5.3,1) {$\beta$};

\draw [dashed][red, line width=0.8] (2.15,0) ellipse (2.15 and
0.8); \node [above right] at (-0.1,1) {$\alpha$};

\draw[line width=0.2mm, -{Stealth[blue]}] (4.15,0.818) --
(4.15001,0.818) node[midway, above=0.2pt] {};

\draw[line width=0.2mm, -{Stealth[reversed,red]}] (2.15,0.818) --
(2.15003,0.818) node[midway, above=0.2pt] {};

\node [below right] at (2.3,0) {$z_2$}; \fill   (2.8,0) circle
(0.5mm);

\node [below right] at (5.3,0) {$z_3$}; \fill   (5.8,0) circle
(0.5mm);

\node [below right] at (0.3,0) {$z_1$}; \fill   (0.8,0) circle
(0.5mm);

\end{tikzpicture}
\caption{The $\beta$ - cycle (blue) and the vanishing $\alpha$ -
cycle (red)}
\end{figure}
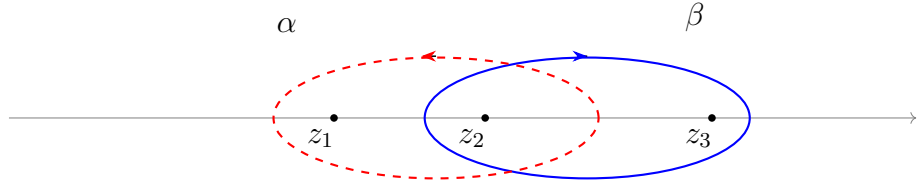

Before stating the next lemma we recall that the functions
\begin{align*}
 F (x; \kappa) &= \int_0 ^x \frac{d t}{\sqrt{(1-t^2)(1-\kappa^2 t^2)}} = \int_0 ^{\theta} \frac{d t}{\sqrt{1 - \kappa^2 \sin^2 t}} =: F (\theta; \kappa), \\
 E (x; \kappa) &= \int_0 ^x \sqrt{\frac{1-\kappa^2 t^2}{1-t^2}} d t = \int_0 ^{\theta} \sqrt{1 - \kappa^2 \sin^2 t} d t = : E (\theta; \kappa), \\
 \Pi (n; x; \kappa) &= \int_0 ^x \frac{d t}{(1-n t^2) \sqrt{(1-t^2)(1-\kappa^2 t^2)}} \\
 &= \int_0 ^{\theta} \frac{d t}{(1-n \sin^2 t)\sqrt{1 - \kappa^2 \sin^2 t}} =: \Pi (n; \theta, \kappa)
\end{align*}
are called incomplete elliptic integrals in the Legendre canonical
form of first, second and third kind, respectively. The number
$\kappa$ is called modulus, $n$ is said to be characteristic.

For the particular value $x=1$, we have the notations $K
(\kappa):= F(1; \kappa), \, E (\kappa):= E(1; \kappa)$ and $\Pi
(n, \kappa):= \Pi (n; 1, \kappa)$, which are known as complete
elliptic integrals of first, second and third kind, respectively.

\begin{lm}
\label{lem1}
The action integral (\ref{3.7}) has the following representation
in terms of the standard Legendre's integrals
\begin{equation}
\label{3.8}
 2 \pi  I_1 (h, j_2) = g_1 K (\kappa) + g_2 E (\kappa) + g_3  \Pi (n, \kappa),
\end{equation}
where
$$
g_1 =  \frac{4 (2h + z_1)}{3 \sqrt{z_3 - z_1}}, \quad g_2 =
\frac{4 \sqrt{z_3 - z_1}}{3}, \quad g_3 = - \frac{ j_2 ^2}{ z_3
\sqrt{z_3 - z_1}}
$$
and
\begin{equation}
\label{3.9}
\kappa^2 = \frac{z_3 - z_2}{z_3 - z_1}, \qquad n = \frac{z_3 - z_2}{z_3}.
\end{equation}
\end{lm}

{\bf Proof.} We transform $I_1$ as follows
\begin{equation}
\label{3.10}
\frac{4 \pi I_1}{\sqrt{2}} = \oint_{\beta}
\frac{w}{z} d z = \oint_{\beta} \frac{w^2}{z w} d z = - 2
\oint_{\beta} \frac{z^2}{w} d z + 2 \oint_{\beta} \frac{z}{w} d z
+ 2h \oint_{\beta} \frac{d z}{w} - \frac{j_2^2}{2} \oint_{\beta}
\frac{d z}{z w} .
\end{equation}
Remembering that the $\beta$-cycle encircles the interval $[z_2,
z_3]$ we get
\begin{align}
\label{3.11}
2 \pi I_1 &= -2  \int_{z_2} ^{z_3} \frac{z^2 d z}{\sqrt{(z_3 -z)(z-z_2)(z-z_1)}} + 2  \int_{z_2} ^{z_3} \frac{z d z}{\sqrt{(z_3 -z)(z-z_2)(z-z_1)}} \\
& + 2h  \int_{z_2} ^{z_3} \frac{ d z}{\sqrt{(z_3
-z)(z-z_2)(z-z_1)}} - \frac{j_2 ^2}{2}  \int_{z_2} ^{z_3} \frac{ d
z}{z \sqrt{(z_3 -z)(z-z_2)(z-z_1)}}. \nonumber
\end{align}
Now, we can transform these integrals in a standard way to their
Legendre canonical form or use appropriate formulas from Byrd \&
Friedman \cite{ByrdFriedman}. For example,
\begin{align*}
\int_{z_2} ^{z_3} \frac{d z}{z \sqrt{(z_3 -z)(z-z_2)(z-z_1)}} \overset{z=z_3-t^2}{= =} 2 \int_0 ^{\sqrt{z_3-z_2}} \frac{dt}{(z_3 -t^2)\sqrt{(z_3-z_2-t^2)(z_3-z_1-t^2)}} \\
\overset{t=\sqrt{z_3-z_2}u}{= =} \frac{2}{z_3\sqrt{z_3-z_1}}
\int_0 ^1 \frac{d u}{(1- n u^2)\sqrt{(1-u^2)(1-\kappa^2 u^2)}} =
\frac{2}{z_3 \sqrt{z_3 -z_1}} \Pi (n, \kappa).
\end{align*}
(compare with the formula 226.02 in \cite{ByrdFriedman}).

In a similar way we get
$$
\int_{z_2} ^{z_3} \frac{d z}{\sqrt{(z_3 -z)(z-z_2)(z-z_1)}} =
\frac{2}{\sqrt{z_3-z_1}} K (\kappa),
$$
$$
\int_{z_2} ^{z_3} \frac{z d z}{\sqrt{(z_3 -z)(z-z_2)(z-z_1)}} =
\frac{2 z_1}{\sqrt{z_3-z_1}} K (\kappa) + 2 \sqrt{z_3 -z_1} E
(\kappa)
$$
and
$$
\int_{z_2} ^{z_3} \frac{z^2 d z}{\sqrt{(z_3 -z)(z-z_2)(z-z_1)}} =
\frac{2h +4z_1}{3\sqrt{z_3-z_1}}K (\kappa) +\frac{4}{3} \sqrt{z_3
-z_1} E (\kappa)
$$
(or use 235.00, 236.01 and 230.01 from \cite{ByrdFriedman}
alternatively). Combining all these expressions we obtain the
needed result (\ref{3.8}).

\hfill $\square$

\begin{lm}
\label{lem2}
The non-trivial action of the champagne bottle near the
focus-focus point $(0, 0)$ as a function of $h$ and $j_2$ has the
expansion
\begin{align}
\label{3.13}
2 \pi I_1 (h, j_2) &= \frac{4}{3} -\pi | j_2 | + j_2 \arctan \frac{j_2}{h} + \mathcal{J} (h, j_2) \ln  \frac{16}{\sqrt{h^2 +  j_2^2}}  \\
          &+ h + \frac{17 h^2 + 3 j_2^2}{16} -\frac{h (118 h^6 + 125 h^4 j_2 ^2 + 194 h^2 j_2 ^4 + 39 j_2 ^6)}{64 (h^2 + j_2^2)^2} + \ldots , \nonumber
\end{align}
where the coefficient in front of the logarithm is
\begin{equation}
\label{3.14}
\mathcal{J} (h, j_2) =  h - \frac{1}{8}(3 h^2 + j_2 ^2) + \frac{h}{64} (35 h^2 + 15 j_2^2) -
\frac{1}{1024}(1155 h^4 + 630 h^2 j_2^2 + 35 j_2^4) + \ldots .
\end{equation}
\end{lm}

\noindent
{\bf Proof.} Clearly  when $h, j_2 \to 0$ then $z_1
\rightarrow 0 \leftarrow z_2$ and $z_3 \to 1$, moreover, $\kappa$
and $n$ tend to 1. Hence, we need the expansions of the complete
elliptic integrals in the singular limit $\kappa \to 1$. To
facilitate the calculations, we put $h \to \mu h$ and $j_2 \to \mu
j_2$, where $\mu$ is a formal small parameter and develop in
$\mu$. Eventually, we set $\mu = 1$, but keep in mind that $h$ and
$j_2$ are close to zero.

\vspace{2ex}

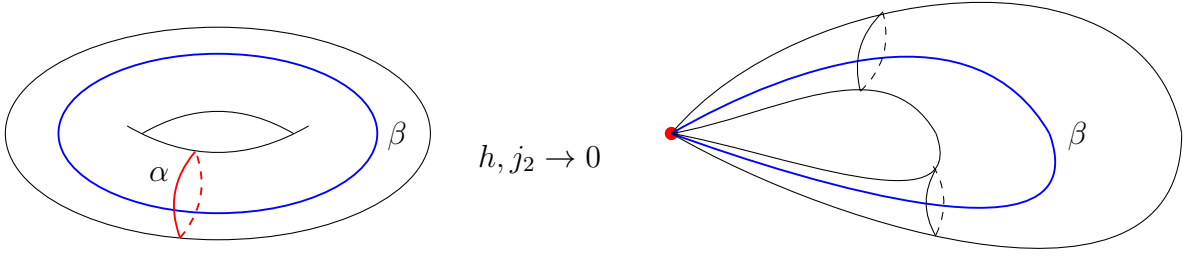
\begin{figure}[ht]
\centering
 \begin{tikzpicture}
  \draw (-1,0) to[bend left] (1,0);
  \draw (-1.2,.1) to[bend right] (1.2,.1);
  \draw[rotate=0] (0,0) ellipse (80pt and 40pt);
 \draw[line width=0.25mm, blue][rotate=0] (0,0) ellipse (60pt and 30pt);
\draw[line width=0.25mm, red] (-0.5,-1.38) to[bend left]
(-0.3,-0.24); \draw[line width=0.25mm, red,dashed] (-0.5,-1.38)
to[bend right] (-0.3,-0.22);

\node [below left] at (-0.5,-0.3) {$ \alpha$};

\node [below right] at (2.1,0.3) {$ \beta$};

\node [below right] at (3.3,0) {$ h, j_2 \to 0$};

\draw[black] (9.5,-1.36) to[bend left] (9.5,-0.47);
\draw[black,dashed] (9.5,-1.34) to[bend right] (9.41,-0.35);

\fill [red] (6.0,0) circle[radius=2.5pt];

 \draw [black  ] (6,0) to[out=50, in=100] (12.75,0);
\draw [black ] (6,0) to[out=-30, in=-90] (12.75,0);

 \draw [line width=0.25mm, blue ] (6,0) to[out=30, in=120] (11,0);
\draw [line width=0.25mm, blue] (6,0) to[out=-20, in=-70] (11,0);

 \draw [black] (6,0) to[out=9.5, in=120] (9.5,0);
\draw [black] (6,0) to[out=-9.5, in=-70] (9.5,0);

\draw[black] (8.5,0.56) to[bend left] (8.8,1.61);
\draw[black,dashed] (8.5,0.56) to[bend right] (8.8,1.59);

\node [below right] at (11.1,0.3) {$ \beta$};

\end{tikzpicture}
\caption{The singular fibre (pinched torus) at the focus-focus
point.}
\end{figure}

\vspace{2ex}

The roots of $P (z)$ have the following expansions
\begin{eqnarray}
\label{3.15}
z_1 =&  - \mu \frac{\sqrt{h^2 +  j_2 ^2} + h}{2}  + \mu^2 c_2 + \ldots, \nonumber \\
z_2 =&  + \mu \frac{\sqrt{h^2 +  j_2 ^2} - h}{2}  + \mu^2 b_2 + \ldots, \\
z_3 =&  1 + \mu h - \mu^2 (h^2 + \frac{j_2 ^2}{4}) + \ldots, \nonumber
\end{eqnarray}
where
$$
b_2 = \frac{(4 h^2 + j_2 ^2) \sqrt{h^2 + j_2^2} -4 h^3 - 3 h
j_2^2}{8 \sqrt{h^2 + j_2^2}}, \qquad c_2 = \frac{(4 h^2 + j_2 ^2)
\sqrt{h^2 + j_2^2} + 4 h^3 + 3 h j_2^2}{8 \sqrt{h^2 + j_2^2}}.
$$
 The complementary modulus $\kappa' $ becomes
$$
\kappa'^2 = 1 - \kappa^2 =  \mu d_1 + \mu^2 d_2 + \ldots,
$$
where
$$
d_1 = b_1-c_1, \qquad  d_2 = b_2 - c_2 + (b_1 - c1)(c_1 - a_1).
$$
Denote $ \Lambda := \ln \frac{\kappa'^2}{16}$. The above formula for $\kappa' $ enables us to find that
\begin{align}
\label{La}
 \Lambda   =  \ln \frac{\mu d_1}{16} + \frac{d_2}{ d_1} \mu  +  \frac{2 d_3 d_1 - d_2^2}{2 d_1^2 }\mu^2 + \ldots .
\end{align}
The following expansions near $\kappa = 1$ ($\kappa' =0$) for $K
(\kappa)$ and $E (\kappa)$ can be found in Cayley \cite{Cayley}
\begin{equation}
\label{3.16}
K (\kappa) = -\frac{1}{2}\Lambda \Big[1 + \frac{1}{4}
\kappa'^2 + \frac{9}{64} (\kappa'^2)^2 + \frac{25}{256}
(\kappa'^2)^3 + \ldots \Big] - \frac{1}{4} \kappa'^2 -
\frac{21}{128} (\kappa'^2)^2 - \frac{185}{1536} (\kappa'^2)^3 +
\ldots
\end{equation}

\begin{equation}
\label{3.17}
E (\kappa) = -\Lambda \Big[\frac{1}{4} \kappa'^2 +
\frac{3}{32} (\kappa'^2)^2 + \frac{15}{256} (\kappa'^2)^3 + \ldots
\Big] + 1 - \frac{1}{4} \kappa'^2 - \frac{13}{64} (\kappa'^2)^2 -
\frac{9}{64} (\kappa'^2)^3 + \ldots
\end{equation}
Since $\kappa^2 < n < 1$ (positive circular case), the elliptic
integral $\Pi (n, \kappa)$ reduces to Heuman's Lambda function $
\Lambda_0  $ (see Byrd \& Friedman \cite{ByrdFriedman}, p. 227)
\begin{equation}
\label{3.18}
\Pi (n, \kappa) = K (\kappa) + \frac{\pi}{2}
\sqrt{\frac{n}{(n - \kappa^2)(1-n)}} \left(1 - \Lambda_0 (\theta, \kappa) \right).
\end{equation}
Then,
$$
g_3 \Pi (n, \kappa) = -\frac{j_2 ^2}{z_3 \sqrt{z_3-z_1}} \mu^2 -\pi | j_2 | \mu + \pi | j_2 | \mu \Lambda_0 (\theta, \kappa).
$$
The function $\Lambda_0 (\theta, \kappa)$ is expressed via incomplete elliptic integrals $F (\theta,
\kappa)$ and $E (\theta, \kappa)$ as
\begin{equation}
\label{3.19}
\Lambda_0 (\theta, \kappa) = \frac{2}{\pi} \Big[K
(\kappa) E (\theta, \kappa') - (K (\kappa) - E (\kappa)) F
(\theta, \kappa')  \Big]
\end{equation}
with
\begin{equation}
\label{3.20}
\theta = \arcsin \sqrt{\frac{1 - n}{(\kappa')^2}}.
\end{equation}

Taking into account the expansions of $K$ and $E$, we obtain for
the Heuman's Lambda function (see  also \cite{A})
\begin{align}
\label{3.21}
\Lambda_0 (\theta, \kappa) &= \frac{2}{\pi} \theta -
\sin \theta \cos \theta \Big[\frac{\kappa'^2}{2\pi} +
\left(\frac{13}{32 \pi} + \frac{3 \sin^2 \theta}{16 \pi}  \right)
(\kappa'^2)^2 + \ldots \Big] \\ \nonumber & - \Lambda \sin \theta
\cos \theta \Big[\frac{\kappa'^2}{2\pi} + \left( \frac{3}{16 \pi}
+ \frac{ \sin^2 \theta}{8 \pi} \right) (\kappa'^2)^2 + \ldots
\Big].
\end{align}
The leading term in the expansion of $2 \theta$ is
$$
2 \theta = 2 \arcsin \sqrt{\frac{\sqrt{h^2 + j_2^2}-h}{2 \sqrt{h^2 +  j_2^2}}  }  +\ldots
$$
can be transformed using the relation \cite{DI}
$$
2 \arcsin \delta = \arctan \gamma \, \Rightarrow \, \gamma = 2 \frac{\sqrt{1-\delta^2} \delta}{1 - 2 \delta^2}
$$
into the form
\begin{equation}
\label{3.22}
2 \theta = \arctan \frac{| j_2 |}{h} +\ldots .
\end{equation}
Finally, combining all expansions given above we get the desired result.

\hfill $\square$

{\bf Remark 2.} As it is noticed by Dullin \cite{Dullin} the
coefficient $\mathcal{J} (h, j_2)$ in front of the logarithm is
also a complete elliptic integral, but along the $\alpha$ - cycle
(imaginary cycle), which is vanishing as $h, j_2 \to 0$. To
calculate it,  the residue theorem can be used. Recall that the $\alpha$ - cycle
encircles the interval $[z_1, z_2]$ along which $w^2 \leq 0$, twice at that. Then
$$
\oint_{\alpha} \frac{w}{z} d z = 2 \pi i \, res_{z=0} \, \frac{w}{z} .
$$
Since only the values of the integrand $\frac{w}{z}$ close to $(h,
j_2)=(0,0)$ matter, we formally scale $h \to \mu h, j_2 \to \mu
j_2$ and find the expansion of $\frac{w}{z}$ with respect to $\mu
<< 1$. As a result we get
$$
res_{z=0} \frac{w}{z} = \frac{1}{\sqrt{2}} \Big[h \mu - \frac{1}{8} (3h^2 + j_2^2) \mu^2 + \frac{h}{64}(35 h^2 + 15 j_2^2 ) \mu^3
 - \frac{1}{1024}(1155 h^4 + 630 h^2 j_2^2 + 35 j_2^4) \mu ^4 \Big] + \ldots .
$$
After putting as before $\mu = 1$ and keeping in mind that $h$ and
$j_2$ are small, we obtain exactly the inverse of the Birkhoff normal form (\ref{invnf})
\begin{align}
\label{3.23}
\mathcal{J} (h, j_2) &=  h - \frac{(3 h^2 + j_2 ^2)}{8} + \frac{h(35 h^2 + 15 j_2^2)}{64} \\
 &- \frac{1155 h^4 + 630 h^2 j_2^2 + 35 j_2^4}{1024} + \ldots := \frac{\sqrt{2}}{2 \pi i} \oint_{\alpha} \frac{w}{z} d z . \nonumber
\end{align}

\section{Semi-global symplectic invariant}

In our case the critical value of the simple focus-focus point is $(h, j_2) = (0, 0)$. By our assumptions, the singular fibre
$F^{-1} (0)$ contains only one critical point $m$ and the component of $F^{-1} (0)$ containing $m$ is compact, see Fig. 3.
In a neighborhood of the focus-focus point, we can consider a momentum map $J = (J_1, J_2)$ consisting of momenta $J_1$ and $J_2$
of the quadratic normal form. In this neighborhood, we may assume that $H$ and the second integral (in this case $J_2$) are functions of
$J_1$ and $J_2$. This is a general result due to Eliasson \cite{Eliasson} and we demonstrated it in practice by constructing the
normal form $\overline{H}$ via near-identity transformations. So, in this neighborhood we can use $\overline{H}$ instead $H$. Let the point
$j= (j_1, j_2) \in \mathbb{R}^2$ be the image of the momentum map $J = (J_1, J_2)$. It can be identified by the complex number
$\hat{j} = j_1 + i j_2$. Therefore, we can use $(j_1, j_2)$ as coordinates instead $(h, j_2)$ in the considered neighborhood of $(0, 0)$.

According to \cite{Ngoc} the non-trivial action near a non-degenerate focus - focus point can be written as follows
\begin{equation}
\label{5.1}
2 \pi I_1 (j_1, j_2) = 2 \pi I_{10} - \mathrm{Re} (\hat{j} \ln \hat{j} - \hat{j}) + S (j_1, j_2),
\end{equation}
where $I_{10}$ is a constant and $S (j_1, j_2)$ is the semi-global symplectic invariant.
In our case we have
\begin{theorem}
\label{Th4}
The nontrivial action $I_1$ of the champagne bottle near focus-focus point (0, 0) is given by
$$
2 \pi I_1 = \frac{4}{3} - \pi | j2 | + j_2 arg \hat{j} - j_1 \ln | \hat{j} | + j_1 + S (j_1, j_2),
$$
where
\begin{equation}
\label{5.2}
S (j_1, j_2) = j_1 \ln 16 + \frac{1}{16} (17 j_1 ^2 + 3 j_2 ^2) - \frac{j_1}{128} (125 j_1 ^2 + 43 j_2 ^2) + \ldots .
\end{equation}
\end{theorem}

\noindent
{\bf Proof.} We only need to substitute the normal form  (\ref{nf}) into the expansion of the action $I_1$ (\ref{3.13}) having in mind that
$\mathcal{J} (h (j_1, j_2), j_2) = j_1$.

\hfill $\square$

Next, we use the obtained expansions to calculate the important dynamical quantities, such as the period $T$ and the rotation number $W$.

The period of the reduced system with respect to the $S^1$-action induced by $J_2$ is defined as
\begin{equation}
\label{5.3}
T (h, j_2) = 2 \pi \frac{\partial I_1}{\partial h} (h, j_2).
\end{equation}
The period can also be viewed as a function of $j_1, j_2$ near the focus-focus point: $T (j_1,j_2) = T (\overline{H} (j_1, j_2),j_2)$.
Then, the result of Theorem \ref{Th4} yields
$$
T (j_1, j_2) =  2 \pi \frac{\partial I_1 / \partial j_1}{\partial \overline{H} / \partial j_1}
= \frac{- \ln | \hat{j} | + S_1}{\partial \bar{H} / \partial j_1 }
= \frac{ \ln \frac{16}{ | \hat{j} | } + \frac{17}{8} j_1 + O(2) }{1 + \frac{3}{4} j_1 + O(2)},
$$
where $S_k = \partial S / \partial j_k, k = 1, 2$. As expected, the period goes to infinity as $\hat{j} \to 0$, i.e.
when approaching the focus-focus point. More detailed expansion can be obtained merely by di\-fferentiating
(\ref{3.13}) in Lemma \ref{lem2} with respect to $h$.

The rotation number is defined as follows
\begin{equation}
\label{5.4}
W (h, j_2) = - \frac{\partial I_1}{\partial j_2} (h, j_2).
\end{equation}
Clearly, the easiest way to get it is to differentiate (\ref{3.13}) in Lemma \ref{lem2} with respect to $j_2$.
Alternatively, we consider $W$  as a function of $j_1, j_2$. Then, \cite{Dullin}
$$
2 \pi W (j_1, j_2) = -2 \pi \frac{\partial I_1}{\partial j_2} (j_1, j_2) = \pi sgn (j_2) - arg \hat{j} - A \ln | \hat{j} | + A S_1 - S_2,
$$
where $A = A (j) := \partial_{j_2} \overline{H} / \partial_{j_1} \overline{H}$, which  expansion reads as follows
$$
A (j) = \frac{j_2}{4} - \frac{15 j_1 j_2}{32} + \frac{j_2 (135 j_1 ^2 + 10 j_2 ^2)}{128} + \ldots .
$$
From here $W (j_1, j_2)$ can be easily found to be
\begin{align}
\label{5.5}
2 \pi W (j_1, j_2) = \pi sgn (j_2) - arctan \frac{j_2}{j_1} & + \ln \frac{16}{\sqrt{j_1 ^2 + j_2 ^2}}
\left( \frac{j_2}{4} - \frac{15}{32} j_1 j_2 + \frac{j_2}{128} (135 j_1 ^2 +10 j_2 ^2) + \ldots \right) \nonumber \\
& - \frac{3}{8} j_2 + \frac{77}{64} j_1 j_2 - \frac{j_2}{512} (885 j_1 ^2 + 43 j_2 ^2) + \ldots .
\end{align}
Notice that $W$ is odd with respect to $j_2$.

Another important dynamical quantity is the twist $\mathcal{T}$, which is defined as
\begin{equation}
\label{5.6}
\mathcal{T} (h, j_2) =  \frac{\partial W}{\partial j_2} (h, j_2)
\end{equation}
for constant $h$. $\mathcal{T}$ is even with respect to $j_2$. The twist condition $\partial W /\partial j_2 \neq 0$ is equivalent to the iso-energetic non-degeneracy condition
(or Arnold-Moser condition), which ensures the survival of invariant tori under small perturbations for the same energy.
Using the developed above technique Dullin and V\~{u} Ng\d{o}c \cite{DSVN} showed that when the focus-focus point is loxodromic,
there exist tori with vanishing twist for each value of $h$ close to the critical one. In our case the eigenvalues are $\pm \sqrt{2}$,
see Section 3. Nevertheless, the result remains the same.

From \cite{DSVN,Dullin} we have a representation of $\mathcal{T}$
\begin{equation}
\label{5.7}
\mathcal{T} (j_1, j_2) = -A (j) \frac{\partial W}{\partial j_1} + \frac{\partial W}{\partial j_2}.
\end{equation}
Straightforward calculations give that
$$
2 \pi \mathcal{T} (j) = \frac{j_1 (A (j) ^2 -1)}{j_1 ^2 + j_2 ^2} - \frac{2 j_2 A (j)}{j_1 ^2 + j_2 ^2} -\frac{3}{8} + \frac{77}{64}j_1 + O(2)
+ \ln \frac{16}{\sqrt{j_1 ^2 + j_2 ^2}} \left(\frac{1}{4} - \frac{15}{32}j_1 + O(2) \right).
$$
After introducing the mapping $\tilde{\mathcal{T}} (j) = | j |^2 2 \pi \mathcal{T} (j)$ the above expression becomes
$$
\tilde{\mathcal{T}} (j) = (A^2 (j) -1) j_1 - 2 A (j) j_2 + O (| j |^2 \ln | j |).
$$
Since $| j |^2 \ln | j | $ is $C^1$ at the origin, it follows that
$\tilde{\mathcal{T}} (0) = 0, \partial \tilde{\mathcal{T}} / \partial j_1 (0) = -1, \partial \tilde{\mathcal{T}} / \partial j_2 (0) = 0$.
Then $\mathcal{T}^{-1} (0)$ is a $C^1$ curve through origin with a tangent at origin  $j_1 = 0$. In a neighborhood of the origin $j_1 \sim h$,
hence we recover the shape of the curve along which the twist condition is violated near the origin $(h, j_2) = (0, 0)$, see Fig. 1.
There is no need to comment on the Kolmogorov condition, because it is verified in several ways \cite{GG,Rink,DSVN,Zung1}.

\section{Concluding Remarks}

This paper is inspired by the work of Dullin \cite{Dullin}, where he gave the first example of explicit computation of the
semi-global symplectic invariant near focus-focus singularity of the sphe\-rical pendulum.
The spherical pendulum and the system studied here share many common features, for instance, they are generalized semi-toric
systems in the terminology of Pelayo et al. \cite{PRNg}.
Here we make use of Dullin's ideas to calculate the semi-global symplectic invariant for the champagne bottle.
Prior to that we have calculated the Birkhoff normal form
and the non-trivial action near the focus-focus point and after that certain important dynamical quantities.

In a view of the classification of the integrable systems we
compare our result (\ref{5.2}) with the semi-global symplectic invariant of the spherical pendulum
\begin{equation}
\label{6.1}
S_{sp} (j_1, j_2) = j_1 \ln 32 + \frac{3 (j_1^2 + 3 j_2^2)}{32} - \frac{j_1(5 j_1 ^2 + 51 j_2 ^2)} {512}+ \frac{
(55 j_1 ^4 +1230 j_1 ^2 j_2 ^2 + 271 j_2 ^4)} {32 768}+ \ldots .
\end{equation}
Apparently they do not coincide, so we can conclude that the spherical pendulum and the champagne bottle are not symplectically equivalent.

\vspace{3ex}

{\bf Acknowledgements}
\noindent
The work is partially supported by the Research Fund of Sofia
University "St. Kliment Ohridski" under Contract 80-10-30/21.05.2025.

\end{document}